\documentclass[a4paper,11pt]{article}

\usepackage{fullpage}
\usepackage{times}
\usepackage{soul}
\usepackage{url}
\usepackage[utf8]{inputenc}
\usepackage[small]{caption}
\usepackage{graphicx}
\usepackage{amsmath}
\usepackage{booktabs}

\usepackage[ruled,vlined]{algorithm2e}
\SetArgSty{textrm}

\usepackage{enumitem}

\usepackage{libertine}

\usepackage{amssymb,amsmath,amsfonts,amstext,amsthm}

\usepackage[breaklinks=true]{hyperref}
\usepackage[svgnames]{xcolor}
\usepackage[capitalise,nameinlink]{cleveref}
\hypersetup{colorlinks={true},linkcolor={DarkBlue},citecolor=[named]{DarkGreen}}

\usepackage{tikz}  
\usetikzlibrary{arrows}
\usetikzlibrary{patterns,snakes}
\usetikzlibrary{decorations.shapes}
\tikzstyle{overbrace text style}=[font=\tiny, above, pos=.5, yshift=5pt]
\tikzstyle{overbrace style}=[decorate,decoration={brace,raise=5pt,amplitude=3pt}]
\usetikzlibrary{shapes.geometric}
\usetikzlibrary{shapes,positioning}
\usetikzlibrary{calc,positioning}
\usetikzlibrary{shapes.multipart}

\definecolor{cadmiumgreen}{rgb}{0.0, 0.42, 0.24}

\usepackage{bbm}

\newtheorem{theorem}{Theorem}[section]

\newtheorem{lemma}[theorem]{Lemma}

\theoremstyle{definition}

\usepackage{natbib}

\usepackage{mathtools}
\usepackage{xcolor} 

\usepackage{authblk}

\usepackage{subcaption}

\usepackage{mdframed}
\mdfsetup{linewidth=1pt}

\usepackage{xspace}
\usepackage{fancyhdr}
\usepackage{tcolorbox}

\newcommand{\cost}{\text{cost}}

\newcommand{\bx}{\mathbf{x}}

\begin{document}

\allowdisplaybreaks

\title{\bf Settling the Distortion of Distributed Facility Location}

\author[1]{Aris Filos-Ratsikas}
\author[2]{Panagiotis Kanellopoulos}
\author[2]{\\ Alexandros A. Voudouris}
\author[2]{Rongsen Zhang}

\affil[1]{School of Informatics, University of Edinburgh, UK}
\affil[2]{School of Computer Science and Electronic Engineering, University of Essex, UK}

\renewcommand\Authands{ and }
\date{}

\maketitle

\begin{abstract}
We study the {\em distributed facility location problem}, where a set of agents with positions on the line of real numbers are partitioned into disjoint districts, and the goal is to choose a point to satisfy certain criteria, such as optimize an objective function or avoid strategic behavior. A mechanism in our distributed setting works in two steps: For each district it chooses a point that is representative of the positions reported by the agents in the district, and then decides one of these representative points as the final output. We consider two classes of mechanisms: {\em Unrestricted} mechanisms which assume that the agents directly provide their true positions as input, and {\em strategyproof} mechanisms which deal with strategic agents and aim to incentivize them to truthfully report their positions. For both classes, we show tight bounds on the best possible approximation in terms of several minimization social objectives, including the well-known social cost (total distance of agents from chosen point) and max cost (maximum distance among all agents from chosen point), as well as other fairness-inspired objectives that are tailor-made for the distributed setting.
\end{abstract}

\section{Introduction}

The theory of social choice deals with the fundamental question of how to aggregate the opinions or \emph{preferences} of diverse individuals into a collective decision. The quality of such a social decision can be measured in several ways, such as based on axiomatic properties, as is usually the case in economics, or qualitative metrics, an approach mainly stemming from the literature in computer science. The most prominent such metric is that of \emph{distortion} \citep{procaccia2006distortion}, which captures precisely the (in)efficiency of a social choice rule, or a class of such rules that often operate under some restrictions, such as the lack of expressive elicitation of the preferences of the agents. 

The distortion of social choice rules (or \emph{mechanisms}) has been a focal point of research over the past decade, for many different settings; see the recent survey of \citet{survey2021} for an overview. The vast majority of previous works assume a basic setting in which a set of agents have cardinal (i.e., numerical) preferences over a set of possible outcomes (alternatives), and the goal is to quantify the best possible distortion of mechanisms that are given as input limited information about the preferences of the agents (usually rankings that are consistent with the cardinal values) in terms of the \emph{social welfare} objective, the total value of the agents for the chosen outcome \citep{anshelevich2015approximating,boutilier2015optimal,ebadian2022optimized,gkatzelis2020resolving}. 

In many cases, however, the situation is often not that simple. For example, in elections, the agents (now voters) are naturally or artificially partitioned into districts, which elect their representatives, and based on these representatives only a final winner is chosen. More generally, the decision-making process is often \emph{distributed}, in the sense that decisions are first made at a local level, among disjoint sets of agents, and then these decisions are aggregated into a collective outcome. These types of situations are not captured by the simple setting laid out above, and bring forward important challenges and complications when measuring the efficiency of social choice mechanisms.

To capture problems of a more complex nature like the ones mentioned above, \citet{FMV20} initiated the study of the distortion in \emph{distributed social choice}, where decisions are made by mechanisms that operate as follows: The mechanism first chooses a representative alternative for each district according to a local election with the agents of the district as voters, and then chooses one of the representatives as the winner. In their work, \citet{FMV20} considered a setting with agents that have normalized cardinal valuations over the possible outcomes. In follow-up work, \cite{AFV22} studied the same question in the very popular \emph{metric social choice} setting, which has dominated the literature of the distortion over the years. In this setting, agents and alternatives are points on a metric space, and distances capture either physical or ideological distances along different axes. The results of \citet{AFV22} identify mechanisms with low distortion bounds not only for the social welfare but also for several other objectives which are appropriate for the distributed setting, and explore the limitations in the design of distributed mechanisms via (almost) matching lower bounds.  

Importantly, the work of \cite{AFV22} only considers a \emph{discrete} social choice setting, in which there is a finite set of alternatives over which the agents are required to choose. Many real-world problems are better modeled as settings where there is a \emph{continuum} of alternatives (e.g., captured by the line of real numbers). Traditionally, this setting has become known as \emph{facility location} \citep{PT09} and its centralized variant is one of the most well-studied topics in social choice theory; see the recent survey of \citet{chan2021mechanism} for a detailed overview. The distributed variant of the continuous setting was first studied by \citet{FV21}, who provided upper and lower bounds on the distortion of mechanisms for the social cost objective, the sum of costs of all the agents. \citet{FV21} considered two types of mechanisms: (a) mechanisms that are only constrained by the fact that they operate in a distributed environment, and (b) mechanisms that are also constrained to be \emph{strategyproof}, i.e., they do not provide incentives to the agents to lie about their preferences. While for the latter case the authors identified the mechanisms with the best possible distortion, for the former case they only managed to show that the distortion lies in the interval $[2,3]$ leaving open the question of whether a mechanism with distortion $2$ is actually possible. Our first contribution is to settle this open question in the affirmative: we design a novel mechanism for distributed facility location with a distortion of $2$ for the social cost objective. 

Besides the social cost objective, \cite{AFV22} identified three more objectives in the discrete setting  which are particularly meaningful for the distributed setting, namely the maximum cost of any agent within any district, the maximum of the sum of costs of the agents in each district, and the sum of the maximum costs of the agents in each district. Following a similar approach, we study these four objectives in the continuous setting and provide upper and lower bounds on the distortion of distributed mechanisms, both with and without the strategyproofness requirement. All of our bounds are \emph{tight}, meaning that we completely settle the distortion of distributed facility location on the real line. We highlight the distributed facility location setting that we focus on, as well as our results, in more detail below.

\subsection{Setting and results}
We consider a facility location setting with a set of {\em agents} that are positioned in the line of real numbers and are partitioned into disjoint {\em districts}. A {\em distributed mechanism} takes as input the positions of the agents and outputs a single point of the line where a public facility is to be located. This decision is made as follows: For each district, the mechanism chooses a location that is representative of the positions of the agents therein. Afterwards, it chooses the output to be one of the locations that represent the districts. The mechanism is distributed in the sense that the choice of the representative location of each district depends only on the positions reported by the agents that belong to the district. 

We design deterministic distributed mechanisms that satisfy various criteria of interest and achieve the best possible distortion bounds. First, we aim to design distributed mechanisms to approximately optimize social objectives that are functions of the distances between the chosen locations and the positions of the agents. Following the work of \cite{AFV22}, we focus on the following objectives: 
\begin{itemize}
    \item The total distance of the agents ({\em social cost}).
    \item The maximum distance among all agents ({\em max cost}).
    \item The total maximum agent distance in each district ({\em Sum-of-Max cost}) 
    \item The maximum total agent distance in each district ({\em Max-of-Sum cost}).
\end{itemize}

To account for the possibly different sizes of the districts, \citet{AFV22} considered the average total distance of the agents in their objectives (in and over the districts) instead of the sum of distances (i.e., they studied the Average-of-Average, Average-of-Max, and Max-of-Average objectives). To clearly demonstrate the arguments in the proofs without overcomplicating the notation, in this paper we focus exclusively on symmetric districts that contain the same number of agents, in which case the average is equivalent to the sum. Similarly to the case of \citet{AFV22}, our proofs can be easily extended to the case of asymmetric districts by considering the average total distance in the objectives; for the Max cost no change is required at all.

Our first contribution is the design of a novel mechanism that achieves a distortion of $2$ for the social cost; as mentioned above, this settles a question left open in the work of \citet{FV21} in the affirmative, matching their lower bound of $2$. For the remaining objectives we provide mechanisms as well as lower bounds establishing that these mechanisms achieve the best possible distortion. The precise bounds are shown in the first column of Table~\ref{tab:results}. Quite interestingly, and perhaps unexpectedly, our mechanism for the Sum-of-Max objective is \emph{optimal}, that is, it achieves a distortion of $1$. This demonstrates that for this particular objective, the distributed nature of the decision making does not influence the quality of the decision at all, and stands in contrast to the results of \citet{AFV22} for the same objective in the discrete setting.

Next, we consider strategyproof mechanisms, i.e., mechanisms that do not incentivize the agents to misreport their locations. This type of mechanisms were considered by \citet{FV21} who settled their distortion for the social cost. For the remaining three objectives, strategyproof mechanisms have not been previously studied, not even in the discrete setting of \citet{AFV22}. We show tight bounds by carefully composing centralized statistics mechanisms for choosing the district representatives and the final location; in particular, depending on the objective at hand, we appropriately choose the values of two parameters $p$ and $q$ to define mechanisms that work by choosing the position of the $q$-th agent in a district as its representative, and then select the $p$-th representative as the output location. Our results for strategyproof mechanisms are shown in the second column of Table~\ref{tab:results}.

\begin{table}[t]
    \centering
    \begin{tabular}{c | c | c}
                & Unrestricted                  & Strategyproof \\\hline
    Social cost & $2$ (Section~\ref{sec:social-cost})     & $3^\star$ \\
    Max cost    & $2$ (Section~\ref{sec:max-cost})     & $2$ (Section~\ref{sec:max-cost})   \\
    Sum-of-Max  & $1$ (Section~\ref{sec:sum-of-max:unrestricted})     & $1+\sqrt{2}$  (Section~\ref{sec:sum-of-max:sp})  \\
    Max-of-Sum  & $2$ (Section~\ref{sec:max-of-sum:unrestricted})     &  $1+\sqrt{2}$ (Section~\ref{sec:max-of-sum:sp})  \\\hline
    \end{tabular}
    \caption{Overview of our tight distortion bounds for deterministic distributed mechanisms. The bound of $3$ for the social cost and the class of strategyproof mechanisms marked with a $\star$ is due to \citet{FV21}.}
    \label{tab:results}
\end{table}

\subsection{Related work}
The distortion was originally defined by \citet{procaccia2006distortion} to quantify the loss in social welfare due to social choice mechanisms having access only to preference rankings over the possible outcomes, rather than to the complete cardinal structure of the preferences. The distortion of mechanisms has been studied for several social choice problems, including single-winner voting, multi-winner voting, participatory budgeting, and matching in both the normalized utilitarian setting~\citep{boutilier2015optimal,caragiannis2017subset,benade2017participatory,ebadian2022optimized,Aris14}, as well as the metric setting~ \citep{anshelevich2015approximating,anshelevich2017randomized,CSV22,charikar2022randomized,gkatzelis2020resolving,kempe2020communication,kempe22veto}. Recently, the notion of the distortion has been more broadly interpreted as capturing the deterioration of an aggregate objective due to limited information, giving rise to works on communication complexity~\citep{mandal2019thrifty,mandal2020optimal}, query complexity~\citep{amanatidis2020peeking,amanatidis2021matching,amanatidis2022don,ma2020matching}, and other tradeoffs between information and distortion~\citep{abramowitz2019awareness}, as well as the distortion of distributed mechanisms that we study in the present paper~\citep{FMV20,AFV22,FV21}. We refer the reader to the survey of \cite{survey2021} for a detailed exposition.

The literature on strategyproof facility location is also rather extensive. \citet{PT09} were the first to study strategyproof facility location problems as part of their approximate mechanism design agenda. Since then, several variants of the problem have been proposed and studied, including settings in which there are several facilities to locate~\citep{lu2009tighter,lu2010asymptotically,fotakis2016strategyproof}, the space of possible locations is restricted~\citep{feldman2016voting,serafino2016,KVZ22}, the agents have heterogeneous preferences over the facilities~\citep{anastasiadis2018heterogeneous,feigenbaum2015strategyproof,duan2019heterogeneous}, only some of the available facilities can be located~\citep{DFV22,elkind2022approval}, or the aim is to optimize different objectives~\citep{DBLP:conf/aaai/Filos-RatsikasL15,qcai2016minimaxenvy,Zhou2022group-fair}.
We refer the reader to the survey of \citet{chan2021mechanism} for more details. 

\section{Preliminaries}
An instance of our problem is a tuple $I=(N,\bx,D)$, where
\begin{itemize}
\item $N$ is a set of $n$ {\em agents}.
\item $\bx = (x_i)_{i \in N}$ is a vector containing the {\em position} $x_i \in \mathbb{R}$ of agent $i$ on the line of real numbers.
\item $D = \{d_1, ..., d_k\}$ is a set of $k \geq 1$ {\em districts}. Each district $d \in D$ contains a set $N_d \subseteq N$ of agents such that $N_d \cap N_{d'}=\varnothing$ and $\bigcup_{d \in D}N_d = N$. We mainly focus on the case of {\em symmetric} districts so that $|N_d| :=\lambda := n/k$ for every $d \in D$. 
\end{itemize}
A {\em distributed mechanism} $M$ is used to decide the location of a facility based on the positions reported by the agents and the composition of the districts. In particular, given an instance $I$, a distributed mechanism works by implementing the following two steps:
\begin{itemize}
\item Step 1: For each district $d \in D$, using only the positions of the agents in $N_d$, the mechanism chooses a {\em representative location} $y_d \in \mathbb{R}$ for the district.
\item Step 2: Given the representative locations of the districts, the mechanism outputs a single location $M(I) \in \{y_d\}_{d \in D}$ as the {\em winner}.
\end{itemize}
If a location $z$ is chosen, then the distance $\delta(x_i,z) = |x_i-z|$ between the position $x_i$ of agent $i$ and $z$ is the {\em individual cost} of agent $i$ for $z$. 

\subsection{Social objectives and strategyproofness}
We want to design mechanisms that output locations which are efficient according to a social objective. 
Let $z \in \mathbb{R}$ be any location. We consider the following four social minimization objectives:
\begin{itemize}
\item The {\em Sum} cost (or {\em social cost}) of location $z$ is the total individual cost of all agents for $z$: 
$$\sum_{i \in N} \delta(x_i,z) = \sum_{d \in D} \sum_{i \in N_d} \delta(x_i,z).$$

\item The {\em Max} cost of location $z$ is the maximum individual cost over all agents for $z$:  
$$\max_{i \in N} \delta(x_i,z) = \max_{d \in D} \max_{i \in N_d} \delta(x_i,z).$$

\item The {\em Sum-of-Max} cost of location $z$ is the sum over each district of the maximum individual cost therein:  
$$\sum_{d \in D} \max_{i \in N_d} \delta(x_i,z).$$

\item The {\em Max-of-Sum} cost of location $z$ is the maximum over each district of the total individual cost therein:  
$$\max_{d \in D} \sum_{i \in N_d} \delta(x_i,z).$$
\end{itemize}
To simplify our notation, whenever the social objective is clear from context, we will use $\cost(z|I)$ to denote the cost of $z \in \mathbb{R}$ according to the objective function at hand in instance $I$.  

Another goal is to design mechanisms that are resilient to strategic manipulation, that is, they do not allow the agents to unilaterally affect the outcome in their favor (i.e., lead to a location with smaller individual cost) by reporting false positions. Formally, a mechanism is {\em strategyproof} if for any pair of instances $I = (N,(\bx_{-i},x_i),D)$ and $J = (N, (\bx_{-i},x_i'),D)$ that differ in the position of a single agent $i$, it holds that $\delta(x_i,M(I)) \leq \delta(x_i,M(J))$.

\subsection{Distortion of mechanisms}
The {\em distortion} of a distributed mechanism $M$ with respect to some social objective (which defines the cost of each possible location) is the worst case (over all instances) of the ratio between the cost of the location chosen by the mechanism and the minimum cost of any location: 
\begin{align*}
    \sup_{I=(N,\bx,D)} \frac{\cost(M(I)|I)}{\min_{z \in \mathbb{R}}\cost(z|I)}
\end{align*}
By definition, the distortion is always at least $1$. When the numerator is positive and the denominator is (extremely close to) $0$, we will say that the distortion is unbounded. Our goal is to design distributed mechanisms that have an as low distortion as possible with respect to the social objectives defined above. We will consider both unrestricted mechanisms which assume that the agents act truthfully, as well as strategyproof mechanisms which aim to avoid strategic manipulations. 

\subsection{Useful observations} \label{useful}
Before we proceed with the presentation of our main technical results in the upcoming sections, we first state some useful properties. The bounds on the distortion of some of our mechanisms will follow by characterizing worst-case instances, and for that we will need the inequality
\begin{align} \label{eq:main-inequality}
    \frac{\alpha +\gamma}{\beta + \gamma} < \frac{\alpha}{\beta},
\end{align}
which holds for any $\alpha > \beta\geq 0$ and $\gamma > 0$.

\citet{FV21} observed that any distributed mechanism with finite distortion with respect to the social cost (sum objective) must be cardinally unanimous. We extend this result by showing this is true for any of the social objectives we consider in this paper. Formally, a mechanism is {\em cardinally-unanimous} if it chooses the representative location of a district to be $z$ whenever all agents in the district are positioned at $z$. 

\begin{lemma} \label{lem:unanimous}
Any distributed mechanism that achieves finite distortion with respect to any social objective $F \in \{\text{Sum}, \text{Max}, \text{Sum-of-Max}, \text{Max-of-Sum}\}$ must be cardinally-unanimous. 
\end{lemma}

\begin{proof}
Let $M$ be a distributed mechanism that is {\em not} cardinally-unanimous. Consequently, there must exist a location $z$ such that when all the agents of a district are positioned at $z$, the mechanism decides the representative location of the district to be some $y \neq z$. Now, consider an instance in which all agents (no matter which district they belong to) are positioned at $z$. Given the behavior of the mechanism, $y$ is the representative location of all districts, and thus it must be the winner. However, $\cost(z) = 0$ and $\cost(y) > 0$ for any social objective $F$, and thus the distortion is unbounded. So, to achieve finite distortion, any mechanism must be cardinally-unanimous. 
\end{proof}

We next show that each member of a class of intuitive distributed mechanisms is strategyproof. Let $p \in [k]$ and $q \in [\lambda]$. The $p$-Statistic-of-$q$-Statistic mechanism first chooses the representative location of each district to be the position of the $q$-th ordered agent therein, and then outputs the $p$-th ordered representative location as the winner. For example, if $p = \lfloor (k+1)/2 \rfloor$ and $q = \lfloor (\lambda+1)/2 \rfloor$, the mechanism selects the position of the (leftmost) median agent in each district to be its representative location and then selects the (leftmost) median representative location as the winner. All strategyproof mechanisms that achieve the best possible distortion for the various social objectives we consider are members of this class. The next lemma shows that any such mechanism is strategyproof, and will allow us to only focus on bounding the distortion in the next sections. 

\begin{lemma} \label{lem:statistic-sp}
For any $p \in [k]$ and $q \in [\lambda]$, the $p$-Statistic-of-$q$-Statistic mechanism is strategyproof. 
\end{lemma}

\begin{proof}
Let $M$ be the $p$-Statistic-of-$q$-Statistic mechanism. 
Consider any instance $I = (N, \bx, D)$ and let $w = M(I)$ be the location chosen by $M$. 
Let $i$ be any agent that belongs to some district $d \in D$ that is represented by $y$. If the position of $i$ is the final winner, then $i$ clearly has no incentive to deviate. So, without loss of generality, assume that the winner is some location $w > x_i$. Observe that to affect the outcome of the mechanism, agent $i$ must first be able to affect the representative of $d$. We distinguish between the following cases.
\begin{itemize}
    \item If $y < x_i$, then agent $i$ would have to report a position $x_i' < y$ to change the representative of $d$, but such a position cannot affect the final winner as the order of representatives remains the same ($w$ would still be at the right of the representative for district $d$).
    \item If $y > w$, then agent $i$ would have to report a position $x_i' > y$ to change the representative of $d$ to $x_i'$. However, this again cannot affect the final winner as the order of representatives remains the same ($w$ would still be at the left of the representative for district $d$). 
    \item If $y \in [x_i, w]$, then agent $i$ could potentially affect the outcome by reporting a position $x_i' > w$ to change the order of representatives, but this would lead to a higher individual cost as the new winner $x_i'$ would be farther away. 
\end{itemize}
Hence, agent $i$ has no incentive to deviate, thus proving that the mechanism is strategyproof. 
\end{proof}


\section{Social cost} \label{sec:social-cost}
We begin with the social cost (Sum) objective. In previous work, \citet{FV21} showed that the {\sc Median-of-Medians} mechanism (that is, the {\sc $\lfloor \lambda/2 \rfloor$-Statistic-of-$\lfloor k/2 \rfloor$-Statistic} mechanism) has distortion at most $3$, and this is best possible strategyproof mechanism. For the class of unrestricted mechanisms, they showed a lower bound of $2$, thus leaving a gap between $2$ and $3$. Here, we complete the picture by showing a tight bound of $2$ for unrestricted mechanisms. We do this by considering the {\sc Median-of-TruncatedAvg} mechanism which works as follows: For each district, the mechanism considers a set of $\lambda/2$ agents ranging from the $(\lambda/4+1)$-th leftmost to the $(3\lambda/4)$-th leftmost\footnote{For simplicity, we present the mechanism assuming that the number of agents in each district is a multiple of $4$; extending the description of the mechanism and the proof is straightforward.}, and chooses their average as the representative location of the district. Then, it chooses the median representative location as the final location. See Mechanism~\ref{mech:MoTA} for a detailed description.

\newcommand\mycommfont[1]{\normalfont\textcolor{blue}{#1}}
\SetCommentSty{mycommfont}
\begin{algorithm}[ht]
\SetNoFillComment
\caption{\sc Median-of-TruncatedAvg}
\label{mech:MoTA}
\For{each district $d\in D$}
{
    $S_d := \{i \in N_d: i \text{ is at least the } (\lambda/4+1) \text{-th and at most the } (3\lambda/4)\text{-th leftmost agent}\}$\;
    $y_d := \frac{\sum_{i \in S_d}x_i}{|S_d|}$\;
}
\Return $w := \text{Median}_{d \in D}\{y_d\}$\; 
\end{algorithm}

To bound the distortion of {\sc Median-of-TruncatedAvg}, we will characterize the structure of worst-case instances, where the distortion of the mechanism is maximized and is strictly larger than $1$. Let $w$ be the location chosen by the mechanism when given as input a worst-case instance, and denote by $o$ the optimal location; since the objective is the social cost, $o$ is the position of the median agent (or any point between the positions of the median agents in case of an even total number of agents). Without loss of generality, we assume that $w < o$; the case $w>o$ is symmetric. 

We first show that there are cases where, starting from an instance with distortion strictly larger than $1$, moving particular agents to appropriate intervals, leads to new instances that have strictly worse distortion. This transformation will be useful when characterizing the worst-case instances for the mechanism. 

\begin{lemma} \label{lem:SC-moving-lemma}
Let $I$ and $J$ be two instances that differ on the position of a single agent $i$, such that $w$ is the location chosen by the mechanism and $o$ is the optimal location for both instances. The distortion of the mechanism when given $J$ as input is strictly larger than its distortion when given $I$ as input in the following cases:
\begin{itemize}
    \item[(a)] $i$ is positioned at $x_i < o$ in $I$, and at $x_i' \in (x_i, o]$ in $J$;
    \item[(b)] $i$ is positioned at $x_i > o$ in $I$, and at $x_i' \in [o,x_i)$ in $J$.
\end{itemize}
\end{lemma}

\begin{proof}
We want to show that 
\begin{align*}
\frac{\cost(w|I)}{\cost(o|I)} < \frac{\cost(w|J)}{\cost(o|J)}.
\end{align*}
For (a), we have that $\delta(x_i,w) \leq \delta(x_i,x_i') + \delta(x_i',w)$ by the triangle inequality, and also $\delta(x_i,o) = \delta(x_i,x_i') + \delta(x_i',o)$; recall our assumption that $w<o$. So, 
\begin{align*}
\frac{\cost(w|I)}{\cost(o|I)} 
= \frac{\sum_{j \neq i} \delta(x_j,w) + \delta(x_i,w) }{\sum_{j \neq i} \delta(x_j,o) + \delta(x_i,o)}
\leq \frac{\sum_{j \neq i} \delta(x_j,w) + \delta(x_i,x_i') + \delta(x_i',w) }{\sum_{j \neq i} \delta(x_j,o) + \delta(x_i,x_i') + \delta(x_i',o)}.
\end{align*}
Since the distortion of the mechanism when given $I$ as input is strictly larger than $1$ and the distances are non-negative, we can apply Inequality \eqref{eq:main-inequality} with $\alpha = \sum_{j \neq i} \delta(x_j,w) + \delta(x_i',w) $, $\beta = \sum_{j \neq i} \delta(x_j,o) + \delta(x_i',o)$ and $\gamma = \delta(x_i,x_i')$, to obtain
\begin{align*}
\frac{\cost(w|I)}{\cost(o|I)} 
< \frac{\sum_{j \neq i} \delta(x_j,w) + \delta(x_i',w) }{\sum_{j \neq i} \delta(x_j,o) + \delta(x_i',o)}
= \frac{\cost(w|J)}{\cost(o|J)}.
\end{align*}
For (b), observe that $\delta(x_i,w) = \delta(x_i,x_i') + \delta(x_i',w)$ and $\delta(x_i,o) = \delta(x_i,x_i') + \delta(x'_i,o)$. Therefore, the desired inequality again follows by appropriately applying Inequality \eqref{eq:main-inequality}.
\end{proof}

We are now ready to show the following useful structural properties of worst-case instances:
\begin{itemize}
    \item At least $k/2$ districts are represented by $w$ (Lemma~\ref{lem:social-cost-P1});
    \item $o$ can be the only other district representative and all agents in such districts are positioned at $o$ (Lemma~\ref{lem:social-cost-P2}).
\end{itemize}

\begin{lemma} \label{lem:social-cost-P1}
There are no district representatives to the left of $w$. 
\end{lemma}

\begin{proof}
Suppose towards a contradiction that the worst-case instance is such that there is a district $d$ with representative $y < w$. Since $y$ is an average of some agent positions in $d$, there is a set of agents $S \subseteq S_d$ with $x_i \leq w$ for every $i \in S$. We move each agent $i \in S$ to a new position $x_i'$ such that $x_i < x_i' \leq w$ and the truncated average of the agents in $d$ becomes $w$. Clearly, the outcome of the mechanism, as well as the optimal location, remain the same in the new instance; $w$ is still the median representative, and the position of the overall median agent did not change. By Lemma~\ref{lem:SC-moving-lemma}(a) and since $w<o$, moving any agent $i\in S$ to $x_i' \leq w$ leads to a new instance with strictly larger distortion, which contradicts the fact that we start from a worst-case instance. 
\end{proof}

\begin{lemma} \label{lem:social-cost-P2}
Besides $w$, the only other district representative can be $o$, and all agents in such districts are positioned on $o$.
\end{lemma}

\begin{proof}
Suppose towards a contradiction that the worst-case instance $I$ is such that there exists a district $d$ with representative $y \not\in \{w,o\}$. We move every agent $i \in N_d$ from $x_i$ to $x_i'=o$. Hence, the truncated average of the agents in $d$ changes from $y$ to $o$. By Lemma~\ref{lem:social-cost-P1} and since $w$ is the median representative, we have that at least half of the district representatives coincide with $w$. Consequently, the outcome of the mechanism is not affected when we move the agents of $d$. The optimal location also remains the same as the median agent location does not change. By Lemma~\ref{lem:SC-moving-lemma}, the distortion of the new instance we obtain after moving each agent $i$ (irrespective of whether $x_i < o$ or $x_i > o$) is strictly larger than the distortion of instance $I$, contradicting the fact that it is a worst-case instance. 
\end{proof}

We also argue that it suffices to focus on the case where the worst-case instance $I$ consists of just two districts, which will simplify the last part of our proof. 

\begin{lemma}\label{lem:two-districts}
There exists a worst-case instance with two districts, one represented by $w$ and one represented by $o$.
\end{lemma}

\begin{proof}
Consider any worst-case instance, and let $D_w$ and $D_o$ denote the sets of districts represented by $w$ and $o$, respectively. We first argue that $|D_w| = |D_o|$. Note that since $w$ is a median among all representatives, we have $|D_w|\geq |D_o|$. Let us assume that $|D_w|>|D_o|$; we will reach a contradiction by creating a new instance, with strictly larger distortion, that has $|D_w| - |D_o|$ additional districts in which all agents are positioned at $o$. Clearly, in this new instance the mechanism again outputs $w$, while the optimal location remains $o$. Since the agents in the newly added districts contribute $0$ to the optimal cost and strictly greater than $0$ to the social cost of $w$, the distortion is strictly larger. 

Now, since $|D_w| = |D_o|$ and all agents in districts represented by $o$ are positioned at $o$ (by Lemma~\ref{lem:social-cost-P2}), we can without loss of generality limit our focus on worst-case instances with just two districts, one with representative $w$ and one with representative $o$. 
\end{proof}

Having shown that it suffices to consider a worst-case instance with one district $d_w$ that is represented by $w$ and one district $d_o$ in which all agents are positioned at $o$, we now argue about the agent positions in $d_w$. Let $\ell$ and $r$ be the locations of the $(\lambda/4+1)$- and $3\lambda/4$-leftmost agent, respectively, in $d_w$ (i.e., the leftmost and rightmost location among agents in $S_{d_w}$). Clearly, it holds that $\ell\leq w \leq r$. We argue that $r\leq o$, and that all agents not in $S_{d_w}$ are either at $\ell$ or at $o$.

\begin{lemma}\label{lem:r leq o}
In district $d_w$, $r \leq o$. 
\end{lemma}

\begin{proof}
Suppose towards a contradiction that the worst-case instance $I$ is such that $r > o$ in $d_w$, and thus $\ell < w$. Let $L$ be the set of agents in $S_{d_w}$ that are positioned to the left of or at $w$, and $R$ the set of agents in $S_{d_w}$ that are positioned to the right of $o$. By the definition of $w$, for any set $Q \subseteq L$, we have
\begin{align*}
    w &= \frac{2}{\lambda}\left( \sum_{i \in L} x_i + \sum_{i \in R} x_i + \sum_{i \in S_{d_w} \setminus (L \cup R)} x_i \right) \\
      &= \frac{2}{\lambda}\left( \sum_{i \in L} x_i + \sum_{i \in R} (x_i-o) + \sum_{i \in R} o + \sum_{i \in S_{d_w} \setminus (L \cup R)} x_i \right) \\
      &= \frac{2}{\lambda}\left( \sum_{i \in L\setminus Q} x_i + \sum_{i \in Q}\left( x_i + \frac{1}{|Q|}\sum_{j \in R} (x_j-o) \right) + \sum_{i \in R} o + \sum_{i \in S_{d_w} \setminus (L \cup R)} x_i \right).
\end{align*}
Consequently, there must exist a set $L_< \subseteq L$ such that $x_i + \frac{1}{|L_<|}\sum_{j \in R} (x_j-o) \leq w < o$ for every $i \in L_<$; if no such set exists, then the last expression above would be strictly larger than $w$. We obtain a new instance $J$ by moving all agents in $R$ from $x_i$ to $x_i'=o$ and all agents in $L_<$ from $x_i$ to $x_i' = x_i + \frac{1}{|L_<|}\sum_{j \in R} (x_j-o)$. Clearly, $w$ is still the representative of $d_w$ and $o$ the optimal location. By Lemma~\ref{lem:SC-moving-lemma}, since all agents that moved are closer to $o$ in $J$ that in $I$, $J$ must have distortion strictly larger than $I$, a contradiction. 
\end{proof}

\begin{lemma}\label{lem:1/4 agents}
In district $d_w$, the $\lambda/4$ leftmost agents are positioned at $\ell$ and the $\lambda/4$ rightmost agents are positioned at $o$. 
\end{lemma}

\begin{proof}
Assume otherwise and note that all these agents are not in $S_{d_w}$ and, hence, do not affect $w$. By repeatedly applying Lemma~\ref{lem:SC-moving-lemma} and moving each agent $i$ with $x_i<\ell$ to $\ell$ and each agent $i$ with $x_i>r$ at $o$, we reach an instance with strictly larger distortion; a contradiction. 
\end{proof}

We are finally ready to prove the main result of this section. 

\begin{theorem}
The distortion of {\sc Median-of-TruncatedAvg} is at most $2$.
\end{theorem}

\begin{proof}
By Lemmas~\ref{lem:social-cost-P2}, \ref{lem:two-districts} and \ref{lem:1/4 agents}, we have that the $2\lambda$ agents in the worst-case  instance $I$ are distributed on the line as follows: $\lambda/4$ agents are positioned at $\ell$, $5\lambda/4$ agents are positioned at $o$ ($\lambda$ agents from $d_o$ and $\lambda/4$ agents from $d_w$), and $\lambda/2$ agents are positioned in $[\ell,r]$. 
We partition the $\lambda/2$ agents in $S_{d_w}$ into two sets:
$L = \{i \in S_{d_w}: x_i \leq w\}$ and $R = \{i \in S_{d_w}: x_i > w\}$.
Since $r\leq o$ (due to Lemma~\ref{lem:r leq o}) and $w = \sum_{i \in L \cup R} x_i$ (by definition), the optimal cost is
\begin{align}\label{eq:optcost}
    \cost(o|I)&=\frac{\lambda}{4}(o-\ell)+\sum_{i \in L}(o-x_{i})+\sum_{i \in R}(o-x_{i})\nonumber\\
    &=\frac{\lambda}{4}(o-\ell)+\frac{\lambda}{2}(o-w)\nonumber\\
    &=\frac{\lambda}{4}(w-\ell)+\frac{3\lambda}{4}(o-w).
\end{align}
Similarly, the cost of the mechanism is
\begin{align}\label{eq:mechcost}
    \cost(w|I)=\frac{\lambda}{4}(w-\ell)+\sum_{i \in L}(w-x_{i})+\sum_{i \in R}(x_{i}-w)+\frac{5\lambda}{4}(o-w).
\end{align}

By the definition of $w$, $\sum_{i \in L}{(w-x_i)} = \sum_{i\in R}{(x_i-w)}$. Also, again by definition, $|L|\geq 1$. 
If $R = \varnothing$, it must be the case that $\ell = w = r$, and the distortion is at most $5/3$ as Equations~\eqref{eq:optcost} and~\eqref{eq:mechcost} are simplified to $\cost(o|I) = 3(o-w)/4$ and $\cost(w|I) = 5(o-w)/4$, respectively. 
Hence, in the rest of the proof we will assume that $|R| \geq 1$. 

Since $x_i\leq o$ for each agent $i \in R$ and $|L| + |R| = \lambda/2$, we have 
$$\sum_{i \in L}{(w-x_i)} = \sum_{i \in R}{(x_i-w)} \leq |R|(o-w) \Leftrightarrow 
o-w \geq \frac{\sum_{i \in L}{(w-x_i)}}{\lambda/2 - |L|}.$$ 
Similarly, as $x_i\geq \ell$ for each agent $i \in L$, we obtain 
$$\sum_{i \in L}{(w-x_i)} \leq |L|(w-\ell) \Leftrightarrow
w-\ell \geq \frac{\sum_{i \in L}{(w-x_i)}}{|L|}.$$
Let $o-w=\frac{\sum_{i \in L}(w-x_{i})}{\frac{\lambda}{2}-|L|}+\xi_1$ and $w-\ell = \frac{\sum_{i \in L}{(w-x_i)}}{|L|}+\xi_2$, where $\xi_1, \xi_2 \geq 0$. Therefore, Equations \eqref{eq:optcost} and \eqref{eq:mechcost} can be rewritten as
\begin{align*}
    \cost(o|I)&=\frac{\lambda}{4}\left(\frac{\sum_{i \in L}{(w-x_i)}}{|L|}+\xi_2\right)+\frac{3\lambda}{4}\left(\frac{\sum_{i \in L}(w-x_{i})}{\frac{\lambda}{2}-|L|}+\xi_1\right)\\\label{eq:case3b}
    \cost(w|I)&=\frac{\lambda}{4}\left(\frac{\sum_{i \in L}{(w-x_i)}}{|L|}+\xi_2\right)+\frac{5\lambda}{4}\left(\frac{\sum_{i \in L}(w-x_{i})}{\frac{\lambda}{2}-|L|}+\xi_1\right)+2\sum_{i \in L}(w-x_{i}).
\end{align*}
It is not hard to see that, unless the distortion is at most $5/3$ and the claim holds trivially, the ratio is maximized when $\xi_1=\xi_2 = 0$. We can then obtain the following upper bound on the distortion.
\begin{align*}
    \frac{\cost(w|I)}{\cost(o|I)} &\leq \frac{\frac{\lambda}{4|L|}+\frac{5\lambda}{2\lambda-4|L|}+2}{\frac{\lambda}{4|L|}+\frac{3\lambda}{2\lambda-4|L|}}\\
    &\leq 2,
\end{align*}
where the last inequality follows since $\frac{\lambda}{4|L|}+\frac{5\lambda}{2\lambda-4|L|}+2 \leq 2\left(\frac{\lambda}{4|L|}+\frac{3\lambda}{2\lambda-4|L|}\right) \Leftrightarrow (\lambda-4|L|)^{2}\geq 0$.
This concludes the proof. 
\end{proof}


\section{Max cost} \label{sec:max-cost}
We now consider the Max cost objective, for which we show a tight bound of $2$ for both unrestricted and strategyproof mechanisms. We begin with the lower bound.

\begin{theorem}
For Max cost, the distortion of any mechanism (unrestricted or strategyproof) is at least $2$.
\end{theorem}

\begin{proof}
Consider any mechanism and the following instance $I$ with two districts. The agents in the first district are all positioned at $-1$, while the agents in the second district are all positioned at $1$. Due to unanimity (Lemma~\ref{lem:unanimous}), the representatives of the two districts must be $-1$ and $1$, respectively. Hence, the winner is either $-1$ or $1$. However, $\cost(-1|I) = \cost(1|I) = 2$, whereas $\cost(0|I) = 1$, leading to a distortion of $2$. 
\end{proof}

For the upper bound, we consider the {\sc Arbitrary} mechanism, which chooses the representative of each district to be the position of any agent therein, and then chooses any representative as the final winner. See Mechanism~\ref{mech:arbitrary} for a specific implementation of this mechanism using the position of the leftmost agent from each district as the district representative, and then the leftmost representative as the final winner. Clearly, {\sc Arbitrary} is equivalent to some $p$-Statistic-of-$q$-Statistic mechanism depending on the choices within and over districts; for example, the particular implementation of {\sc Arbitrary} as Mechanism~\ref{mech:arbitrary} is equivalent to $1$-Statistic-of-$1$-Statistic. We will now show that the distortion of this mechanism is at most $2$.

\SetCommentSty{mycommfont}
\begin{algorithm}[h]
\SetNoFillComment
\caption{\sc Arbitrary (Leftmost-of-Leftmost)}
\label{mech:arbitrary}
\For{each district $d$}
{
    $y_d := \min_{i \in N_d}\{x_i\}$\;
}
\Return $w := \min_d \{y_d\}$\; 
\end{algorithm}

\begin{theorem}
For Max cost, the distortion of {\sc Arbitrary} is at most $2$.
\end{theorem}

\begin{proof}
Given any instance $I$, let $\ell$ and $r$ denote the positions of the leftmost and the rightmost agent, respectively. Clearly, the optimal location is $o = \frac{r-\ell}{2}$, and thus $\cost(o|I) = \frac{r-\ell}{2}$. On the other hand, the {\sc Arbitrary} mechanism will necessarily return the location of some agent as the winner $w$, and hence $\cost(w|I) \leq r-\ell$; the claim follows. 
\end{proof}

 
\section{Sum-of-Max} \label{sec:sum-of-max}
Here, we focus on the Sum-of-Max objective; recall that for this objective we sum over each district the maximum agent cost therein. For unrestricted mechanisms, we show that, surprisingly, it is possible to achieve a distortion of $1$, whereas, for strategyproof mechanisms, we show a tight bound of $1+\sqrt{2}$. 

\subsection{Unrestricted mechanisms} \label{sec:sum-of-max:unrestricted}
We will show that the {\sc Median-of-Midpoints} mechanism optimizes the Sum-of-Max objective. This mechanism chooses the representative of each district to be the midpoint of the interval defined by the positions of the agents therein, and then chooses the median representative (breaking ties in favor of the leftmost median in case there are two) as the final winner. See Mechanism~\ref{mech:med-of-midpoints} for a detailed description. 

\SetCommentSty{mycommfont}
\begin{algorithm}[h]
\SetNoFillComment
\caption{\sc Median-of-Midpoints}
\label{mech:med-of-midpoints}
\For{each district $d$}
{
    $y_d := \frac{1}{2} \cdot \bigg( \max_{i \in N_d}x_i + \min_{i \in N_d}x_i \bigg)$ \;
}
\Return $w := \text{Median}_{d \in D}\{y_d\}$ \; 
\end{algorithm}

To show the desired bound on the distortion of {\sc Median-of-Midpoints} for Sum-of-Max, we again show some useful properties of worst-case instances. Without loss of generality, we will assume that the mechanism is applied on input a worst-case instance $I$ where the chosen winner $w$ is to the left of the optimal location $o$, that is, $w < o$.  

\begin{lemma}\label{SoM:no-left-midpoints}
There are no midpoints at the left of $w$ or at the right of $o$.
\end{lemma}

\begin{proof}
Due to symmetry, it suffices to prove the first part of the lemma. Suppose towards a contradiction that the worst-case instance $I$ is such that there is a district $d$ with representative (midpoint) $y < w$. So, the leftmost agent $\ell$ in $d$ is positioned at $x_\ell < w$ such that $\max_{i \in N_d} \delta(x_i,w) = \delta(x_\ell,w)$ and $\max_{i \in N_d} \delta(x_i,o) = \delta(x_\ell,o)$. We obtain a new instance $J$ by moving each agent $i \in N_d$ to $x_i' = w$. Hence, the midpoint of $d$ becomes $w$ in $J$. Clearly, the outcome of the mechanism remains the same in the new instance.
We can write the distortion of the mechanism for $I$ as follows:
\begin{align*}
\frac{\cost(w|I)}{\cost(o|I)}
&= \frac{ \sum_{d' \neq d} \max_{i \in N_{d'}} \delta(x_i,w) + \delta(x_\ell,w) }{ \sum_{d' \neq d} \max_{i \in N_{d'}} \delta(x_i,o) + \delta(x_\ell,o)} \\
&= \frac{ \sum_{d' \neq d} \max_{i \in N_{d'}} \delta(x_i,w) + \delta(x_\ell,w) }{ \sum_{d' \neq d} \max_{i \in N_{d'}} \delta(x_i,o) + \delta(x_\ell,w) + \delta(w,o)} \\
&< 
\frac{ \sum_{d' \neq d} \max_{i \in N_{d'}} \delta(x_i,w)  }{ \sum_{d' \neq d} \max_{i \in N_{d'}} \delta(x_i,o) + \delta(w,o)} \\
&=\frac{\cost(w|J)}{\cost(o|J)},
\end{align*}
where the inequality follows by Inequality~\eqref{eq:main-inequality}. The optimal location in $J$ might be a different position $o' \neq o$. Since $\cost(o') \leq \cost(o)$, we have that
\begin{align*}
\frac{\cost(w|I)}{\cost(o|I)} 
< \frac{\cost(w|J)}{\cost(o|J)}
\leq \frac{\cost(w|J)}{\cost(o'|J)}.
\end{align*}
This contradicts the fact that the original instance $I$ is a worst-case instance.
\end{proof}

\begin{lemma}\label{SoM:w-districts}
In each district represented by $w$, all agents are positioned at $w$. 
\end{lemma}

\begin{proof}
The proof is similar to that of Lemma~\ref{SoM:no-left-midpoints}. Suppose towards a contradiction that in the worst-case instance $I$ there is a district $d$ represented by $w$ with at least one agent $i$ positioned at some $x_i \neq w$. Then, since $w$ is the midpoint of the positions of the agents in $d$, the leftmost agent $\ell$ must be positioned at $x_\ell < w$, and thus  
$\max_{i \in N_d} \delta(x_i,w) = \delta(x_\ell,w)$ and $\max_{i \in N_d} \delta(x_i,o) = \delta(x_\ell,o)$. We obtain a new instance $J$ by moving each agent $i \in N_d$ from $x_i$ to $x_i' = w$. Clearly, the outcome of the mechanism remains the same in $J$. We can write the distortion of the mechanism for $I$ as follows:
\begin{align*}
\frac{\cost(w|I)}{\cost(o|I)}
&= \frac{ \sum_{d' \neq d} \max_{i \in N_{d'}} \delta(x_i,w) + \delta(x_\ell,w) }{ \sum_{d' \neq d} \max_{i \in N_{d'}} \delta(x_i,o) + \delta(x_\ell,o)} \\
&= \frac{ \sum_{d' \neq d} \max_{i \in N_{d'}} \delta(x_i,w) + \delta(x_\ell,w) }{ \sum_{d' \neq d} \max_{i \in N_{d'}} \delta(x_i,o) + \delta(x_\ell,w) + \delta(w,o) } \\
&< 
\frac{ \sum_{d' \neq d} \max_{i \in N_{d'}} \delta(x_i,w)  }{ \sum_{d' \neq d} \max_{i \in N_{d'}} \delta(x_i,o) + \delta(w,o)} \\
&=\frac{\cost(w|J)}{\cost(o|J)},
\end{align*}
where the inequality follows by Inequality~\eqref{eq:main-inequality}. The optimal location in $J$ might be a different position $o' \neq o$. Since $\cost(o') \leq \cost(o)$, we have that
\begin{align*}
\frac{\cost(w|I)}{\cost(o|I)} 
< \frac{\cost(w|J)}{\cost(o|J)}
\leq \frac{\cost(w|J)}{\cost(o'|J)}.
\end{align*}
This contradicts the fact that the original instance $I$ is a worst-case instance.
\end{proof}

We are now ready to show that {\sc Median-of-Midpoints} optimizes the Sum-of-Max objective.

\begin{theorem}\label{SoM:2-bound}
For Sum-of-Max, the distortion of {\sc Median-of-Midpoints} is $1$.
\end{theorem}

\begin{proof}
We will show that $\cost(w|I) \leq \cost(o|I)$; to simplify our notation, we drop $I$ from $\cost$ for the rest of this proof. Let $D_w$ be the set of districts represented by $w$, and denote by $\overline{D_w}$ the remaining districts. By Lemma~\ref{SoM:no-left-midpoints} and the fact that $w$ is the median midpoint, we have that $|D_w|\geq |\overline{D_w}|$. Also, by Lemma~\ref{SoM:w-districts}, we know that all agents in the districts of $D_w$ are positioned at $w$. Hence, 
\begin{align*}
\cost(w) &=  \sum_{d \in D_w} \max_{i \in N_d} \delta(x_i,w) + \sum_{d \not\in D_w} \max_{i \in N_d} \delta(x_i,w) \\
&= \sum_{d \not\in D_w} \max_{i \in N_d} \delta(x_i,w).
\end{align*}
We can write the cost of $o$ as follows:
\begin{align*}
\cost(o) &= \sum_{d \in D_w} \max_{i \in N_d} \delta(x_i,o) + \sum_{d \not\in D_w} \max_{i \in N_d} \delta(x_i,o) \\
&= |D_w| \delta(w,o) + \sum_{d \not\in D_w} \max_{i \in N_d} \delta(x_i,o).
\end{align*}
Now, let $d \not\in D_w$ be a district with midpoint $y$, which is such that $y \in (w,o]$ due to Lemma~\ref{SoM:no-left-midpoints}. Let $\ell_d$ and $r_d$ be the leftmost and rightmost agents of district $d$, respectively. We claim that the contribution of $d$ to $\cost(o)$ is the distance $\delta(\ell_d,o)$, and the contribution of $d$ to $\cost(w)$ is the distance $\delta(r_d,w)$. First observe that since $y \in (w,o]$, it cannot be the case that $\ell_d, r_d < w$ or $\ell_d, r_d > o$; in other words, we necessarily have that $w \leq r_d$ and $\ell_d \leq o$. If $\ell_d \leq r_d < o$, then our claim for $\cost(o)$ follows immediately. So, suppose that $r_d > o$, and thus $\ell_d < o$. Then, the inequality $\frac{\ell_d+r_d}{2} \leq o$ due to the fact that $y \in (w,o]$ implies that $r_d - o \leq o - \ell_d$, and thus $\delta(\ell_d,o)$ is the contribution of $d$ to $\cost(o)$. Similarly, if $w \leq \ell_d \leq r_d$, our claim for $\cost(w)$ follows immediately. Thus, we can suppose that $\ell_d < w$, and thus $r_d > w$. The inequality $\frac{\ell_d+r_d}{2} > w$ implies that $w-\ell_d < r_d - w$, and thus $\delta(r_d,w)$ is the contribution of $d$ to $\cost(w)$. By this, we have that 
\begin{align*}
\cost(w) &= \sum_{d \not\in D_w} \delta(r_d,w).
\end{align*}
Furthermore, for each $d \not\in D_w$,
\begin{align*}
\delta(\ell_d,o) 
&= o - \ell_d \\
&= o-r_d + r_d - w + w - \ell_d \\
&= \delta(r_d,w) + (o+w) - (\ell_d + r_d) \\
&\geq \delta(r_d,w) + (o+w) - 2o \\
&= \delta(r_d,w) - (o-w) \\
&= \delta(r_d,w) - \delta(w,o),
\end{align*}
where the inequality follows since $\frac{\ell_d + r_d}{2}\leq o$. Consequently, 
\begin{align*}
\cost(o) &= |D_w| \delta(w,o) + \sum_{d \not\in D_w} \delta(\ell_d,o) \\
&\geq |D_w| \delta(w,o) + \sum_{d \not\in D_w} \bigg( \delta(r_d,w) - \delta(w,o) \bigg) \\
&= (|D_w| - |\overline{D_w}|)\delta(w,o) + \sum_{d \not\in D_w} \delta(r_d,w) \\
&\geq \sum_{d \not\in D_w} \delta(r_d,w) \\
&= \cost(w),
\end{align*}
where the last inequality follows since $|D_w| \geq |\overline{D_w}|$.
\end{proof}

\subsection{Strategyproof mechanisms}\label{sec:sum-of-max:sp}
For strategyproof mechanisms, we will show a tight bound of $1+\sqrt{2}$. We start by showing the lower bound on the distortion of all strategyproof mechanisms. 

\begin{theorem}
For Sum-of-Max, the distortion of any strategyproof mechanism is at least $1+\sqrt{2}-\varepsilon$, for any $\varepsilon>0$.
\end{theorem}

\begin{proof}
Assume towards a contradiction that there is a strategyproof mechanism with distortion strictly smaller than $1+\sqrt{2}-\varepsilon$, for any $\varepsilon>0$. Without loss of generality, we assume that when there are two districts with different representatives, we choose the leftmost as the final winner. We will prove the statement by showing some properties about the behavior of strategyproof mechanisms in particular instances.

\paragraph{Property (P1):} 
We claim that there is a district with two agents such that the mechanism chooses some agent position as the district representative. Consider a district $d$ with one agent positioned at $x$ and one agent positioned at $y > x$. If the mechanism chooses the representative to be $x$ or $y$, then we are done. Otherwise, suppose that the representative is chosen to be some $z \not\in \{x, y\}$. Due to strategyproofness, $z$ must also be the representative of the district $d'$ where any of the two agents has been moved to $z$; otherwise, in the single-district instance consisting of $d'$, the agent that is moved would have incentive to report that she is positioned as in $d$ to change the outcome to $z$.

\paragraph{Property (P2):} 
By Property (P1) there exists a district with two agents such that the mechanism chooses the district representative to be the position of one of the agents; without loss of generality we assume that the agents are positioned at $0$ and $1$. We claim that the representative of this district must be $1$ as otherwise the distortion would be at least $3$. Indeed, suppose otherwise that the representative is $0$, and consider the following instance $I_1$ with two districts:
\begin{itemize}
\item In the first district, there is an agent at $0$ and an agent at $1$. By the above discussion, the representative is $0$.
\item In the second district, there are two agents at $1/2$. Due to unanimity, the representative is $1/2$ (otherwise the distortion would be unbounded due to Lemma~\ref{lem:unanimous}). 
\end{itemize}
Since there are only two districts and two different representatives, the overall winner is $0$. 
But, $\cost(0|I_1) = (1 - 0) + (1/2-0) = 3/2$ and $\cost(1/2 | I_1) = (1 - 1/2) + (1/2-1/2) = 1/2$, leading to a distortion of $3$. 

\paragraph{Property (P3):} Let $\alpha < \beta$ be two (large) integers such that $\beta/\alpha = 1+\sqrt{2}-\delta$, for some arbitrarily small $\delta > 0$. We claim that in instances with $\alpha+\beta$ districts such that $1/2$ is the representative of $\alpha$ districts and $1$ is the representative of $\beta$ districts, the overall winner must be $1$ as otherwise the distortion would be $\beta/\alpha = 1+\sqrt{2}-\delta$. Indeed, suppose that the winner is $1/2$ in such a case, and consider the following instance $I_2$ with $\alpha+\beta$ districts:
\begin{itemize}
\item In $\alpha$ districts, there are two agents at $1/2$.
\item In $\beta$ districts, there are two agents at $1$. 
\end{itemize}
Due to unanimity (Lemma~\ref{lem:unanimous}), the representatives are $1/2$ and $1$, respectively, and the overall winner is $1/2$ by assumption. Then, $\cost(1/2|I_2) = \beta/2$ and $\cost(1|I_2)= \alpha/2$. So, the distortion is at least $\beta/\alpha = 1+\sqrt{2}-\delta$. 

\paragraph{Reaching a contradiction:} Now, we consider the following instance $I_3$ with $\alpha + \beta$ districts:
\begin{itemize}
\item In $\alpha$ districts, there are two agents at $1/2$. Due to unanimity the representative of all these districts is $1/2$.
\item In $\beta$ districts, there is one agent at $0$ and one agent at $1$. By property (P2), the representative of all these districts is $1$.
\end{itemize}
Since $1/2$ is the representative of $\alpha$ districts and $1$ is the representative of $\beta$ districts, by property (P3), the overall winner is $1$. We have that $\cost(1|I_3) = \frac{\alpha}{2} + \beta$ and $\cost(1/2|I_3) = \frac{\beta}{2}$. That is, the distortion is at least $2+\frac{\alpha}{\beta} > 2+\frac{1}{1+\sqrt{2}} = 1+\sqrt{2}$; a contradiction.
\end{proof}

For the tight upper bound, we consider the {\sc $\left(1-1/\sqrt{2}\right)k$-Leftmost-of-Rightmost} mechanism, which chooses the representative of each district to be the position of the rightmost agent therein, and then chooses the $\left(1-1/\sqrt{2}\right)k$-th leftmost representative as the final winner.\footnote{To be precise, the mechanism chooses the $\left\lceil\left(1-1/\sqrt{2}\right)k\right\rceil$-leftmost representative as the winner. To simplify our notation and discussion, we drop the ceiling.} See Mechanism~\ref{mech:LoR} for a detailed description. Clearly, the mechanism is strategyproof as it is an implementation of $p$-Statistic-of-$q$-Statistic with $p=\left(1-1/\sqrt{2}\right)k$ and $q=\lambda$. So, it suffices to show that it achieves a distortion of at most $1+\sqrt{2}$.

\SetCommentSty{mycommfont}
\begin{algorithm}[t]
\SetNoFillComment
\caption{\sc $\left(1-1/\sqrt{2}\right)k$-Leftmost-of-Rightmost}
\label{mech:LoR}
\For{each district $d\in D$}
{
    $y_d := $  rightmost agent\;
}
\Return $w := \left(1-1/\sqrt{2}\right)k$-th leftmost representative\; 
\end{algorithm}

\begin{theorem}
For Sum-of-Max, the distortion of {\sc $\left(1-1/\sqrt{2}\right)k$-Leftmost-of-Rightmost} is at most $1+\sqrt{2}$.
\end{theorem}

\begin{proof}
Consider any instance $I$. Let $w$ be the location chosen by the mechanism, and $o$ the optimal location. For each district $d$, let $i_d$ be the most distant agent from $w$, and $i_d^*$ the most distant agent from $o$. So, $\cost(w|I) = \sum_{d \in D} \delta(i_d,w)$, and $\cost(o|I) = \sum_{d \in D} \delta(i_d^*,o) \geq \sum_{d \in D}\delta(j,o)$ for any agent $j \in N_d$. We consider the following two cases depending on the relative positions of $w$ and $o$.

\paragraph{Case 1: $o < w$.} \ \\
Let $S = \{d \in D: y_d \geq w \}$ be the set of district representatives to the right of $w$. By the definition of $w$, we have that $|S| \geq \frac{k}{\sqrt{2}}$. Since $o < w \leq y_d$ for every $d\in S$ and $y_d \in N_d$, we have that 
\begin{align*}
\cost(o|I) \geq \sum_{d \in S} \delta(y_d,o) \geq |S|\cdot \delta(w,o) \geq \frac{k}{\sqrt{2}} \cdot \delta(w,o) 
\Leftrightarrow k \cdot \delta(w,o) \leq \sqrt{2} \cdot \cost(o|I).
\end{align*}
By the triangle inequality and since $i_d \in N_d$, we have
\begin{align*}
\cost(w|I) = \sum_{d \in D} \delta(i_d,w) 
&\leq \sum_{d \in D} \delta(i_d,o) + \sum_{d \in D}\delta(w,o) \\
&\leq \cost(o|I) + k \cdot \delta(w,o) \\
&\leq (1+\sqrt{2}) \cdot \cost(o|I).
\end{align*}

\paragraph{Case 2: $w < o$.} \ \\
We partition the districts into a set $L$ that includes $\left(1-\frac{1}{\sqrt{2}}\right)k$ districts from the one with the leftmost representative until the one with the $\left(1-\frac{1}{\sqrt{2}}\right)k$-th leftmost representative (that is, $w$), and a set $R$ that includes the remaining districts. By definition, we have that $|R|/|L| = 1+ \sqrt{2}$. For every district $d$, let $\ell_d$ and $r_d$ be the leftmost and rightmost agents in $d$, respectively. 
We make the following observations:
\begin{itemize}
\item For every $d \in L$, since $y_d$ is the rightmost agent of $d$ and $y_d \leq w < o$, it must be the case that $i_d = i_d^* = \ell_d$. Due to the positions of $\ell_d$, $w$ and $o$, we have that $\delta(\ell_d,o) = \delta(\ell_d,w) + \delta(w,o)$.

\item For every $d \in R$, by the triangle inequality, we have that $\delta(i_d,w) \leq \delta(i_d,o) + \delta(w,o)$. Since $\delta(i_d,o) \leq \delta(i_d^*,o)$ by the definition of $i_d^*$, we further have that $\delta(i_d,w) \leq \delta(i_d^*,o) + \delta(w,o)$. 
\end{itemize}
Hence, 
\begin{align*}
\cost(w|I) &= \sum_{d \in D} \delta(i_d,w) = \sum_{d \in L} \delta(\ell_d,w) + \sum_{d \in R} \delta(i_d,w) \\
&\leq \sum_{d \in L} \bigg( \delta(\ell_d,w) + \delta(w,o) \bigg) - |L|\delta(w,o) + \sum_{d \in R} \bigg( \delta(i_d^*,o) + \delta(w,o) \bigg) \\
&= \cost(o|I) + (|R|-|L|) \delta(w,o). 
\end{align*}
Since $y_d \leq w < o$ for every $d \in L$ and $y_d \in N_d$, we have that 
\begin{align*}
\cost(o|I) \geq \sum_{d \in L} \delta(y_d,o) \geq |L| \cdot \delta(w,o) \Leftrightarrow \delta(w,o) \leq \frac{1}{|L|} \cdot \cost(o|I).
\end{align*}
Therefore, we obtain
\begin{align*}
\cost(w|I) \leq \cost(o|I) + \frac{|R|-|L|}{|L|} \cdot \cost(o|I) = \frac{|R|}{|L|}\cdot \cost(o|I) = (1+\sqrt{2})\cdot \cost(o|I),
\end{align*}
as desired. 
\end{proof}


\section{Max-of-Sum}
We now turn our attention to the last objective, Max-of-Sum, which is the maximum over each district of the total individual cost therein. We show a tight bound of $2$ for unrestricted mechanisms and a tight bound of $1+\sqrt{2}$ for strategyproof mechanisms. 

\subsection{Unrestricted mechanisms}\label{sec:max-of-sum:unrestricted}
Since the lower bound of $2$ for the Max cost objective holds even when there is a single agent in each district, it extends to the case of Max-of-Sum as well. For the upper bound, we consider the {\sc Arbitrary-of-Avg} mechanism, which chooses the representative of each district to be the average of the positions of the agents in the district, and then chooses an arbitrary representative (e.g., the leftmost) as the final winner. See Mechanism~\ref{mech:LoAvg} for a detailed description. 

\SetCommentSty{mycommfont}
\begin{algorithm}[h]
\SetNoFillComment
\caption{\sc Arbitrary-of-Avg}
\label{mech:LoAvg}
\For{each district $d\in D$}
{
    $y_d := \frac{\sum_{i \in N_d}{x_i}}{\lambda}$ \;
}
\Return $w := \min_{d \in D}{y_d}$ \; 
\end{algorithm}

\begin{theorem}\label{thm:MoS-unrestricted-upper}
For Max-of-Sum, the distortion of {\sc Arbitrary-of-Avg} is at most $2$.
\end{theorem}

\begin{proof}
Consider any instance $I$. Let $w$ be the location chosen by the mechanism, and $o$ the optimal location; without loss of generality, we assume that $w < o$. Denote by $d^*$ a district that defines the cost of $w$, that is, $d^* \in \arg\max_{d \in D} \sum_{i \in N_d} \delta(x_i,w)$. Also, denote by $d_w$ a district represented by $w$, that is, $w = \frac{1}{\lambda}\sum_{i \in N_{d_w}}x_i \Leftrightarrow \sum_{i \in N_{d_w}}(w-x_i) = 0$. By the triangle inequality, we have that
\begin{align*}
\cost(w|I) 
= \sum_{i \in N_{d^*}} \delta(x_i,w) 
\leq \sum_{i \in N_{d^*}} \delta(x_i,o) + \sum_{i \in N_{d^*}} \delta(w,o)
\leq \cost(o|I) + \lambda \delta(w,o).
\end{align*}
By the definition of $d_w$ and since it consists of $\lambda$ agents, we have that
\begin{align*}
\lambda \delta(w,o) &= \lambda (o-w) = \lambda (o-w) + \sum_{i \in N_{d_w}}(w-x_i) = \sum_{i \in N_{d_w}}(o-x_i) \leq \sum_{i \in N_{d_w}} \delta(x_i,o) \leq \cost(o|I),
\end{align*}
where the inequality follows since $\delta(x_i,o) = o-x_i$ when $x_i \leq o$ and $\delta(x_i,o) = x_i - o \geq o-x_i$ when $x_i \geq o$. Therefore, we obtain that $\cost(w|I) \leq 2\cdot \cost(o|I)$, as desired. 
\end{proof}

\subsection{Strategyproof mechanisms}\label{sec:max-of-sum:sp}
We now turn out attention to strategyproof mechanisms and first show a lower bound of $1+\sqrt{2}$. 

\begin{theorem}
For Max-of-Sum, the distortion of any strategyproof mechanism is at least $1+\sqrt{2}-\varepsilon$, for any $\varepsilon>0$.
\end{theorem}

\begin{proof}
Suppose towards a contradiction that there is a strategyproof mechanism with distortion strictly smaller than 
$1+\sqrt{2}-\varepsilon$, for any $\varepsilon>0$. 

\paragraph{Property (P1):} Consider a district with $(1+\sqrt{2})x$ agents at $0$ and $x$ agents at $1$, where $x$ is an arbitrarily large integer.\footnote{To be precise, since the number of agents must be an integer, we would need to have $\lceil (1+\sqrt{2})x \rceil$ agents at $0$. We simplify our notation by dropping the ceilings, but it should be clear that this does not affect our arguments.} We claim that the mechanism must choose $0$ as the representative of this district as otherwise the distortion would be at least $1+\sqrt{2}$. Indeed, suppose that the representative is some $y \neq 0$. By moving one of the agents at $1$ to $y$, we obtain a new district whose representative must still be $y$; otherwise, in the instance that consists only of this new district, the agent at $y$ would have incentive to misreport her position as $1$, thus leading to the representative (and the final winner) to change to $y$. By induction, we obtain that $y$ must be the representative of the district with $(1+\sqrt{2})x$ agents at $0$ and $x$ agents at $y$. In the instance $I$ that consists of only the latter district, the winner is $y$ with $\cost(y|I) = (1+\sqrt{2})x\cdot |y|$, whereas $\cost(0|I) = x\cdot|y|$, leading to a distortion of at least $1+\sqrt{2}$.

\paragraph{Property (P2):} Consider a district with $x$ agents at $1$ and $(1+\sqrt{2})x$ agents at $2$. We claim that the mechanism must choose $2$ as the representative of this district as otherwise the distortion would be at least $1+\sqrt{2}$. This follows by arguments similar to those for property (P1). 

\paragraph{Reaching a contradiction:}
Consider the following instance $J$ with two districts:
\begin{itemize}
    \item In the first district, there are $(1+\sqrt{2})x$ agents at $0$ and $x$ agents at $1$. 
    \item In the second district, there are $x$ agents at $1$ and $(1+\sqrt{2})x$ agents at $2$.
\end{itemize}
By properties (P1) and (P2), the representatives of the two districts must be $0$ and $2$, respectively, and thus one of these two locations is chosen as the final winner. However, $\cost(0|J) = \cost(2|J) = 2(1+\sqrt{2})x+x$, while $\cost(1|J) = (1+\sqrt{2})x$, leading to a distortion of $2 + \frac{1}{1+\sqrt{2}} = 1+\sqrt{2}$. 
\end{proof}

For the tight upper bound, we consider the {\sc Rightmost-of-$\left(1-1/\sqrt{2}\right)\lambda$-Leftmost} mechanism, which chooses the representative of each district to be the position of the $\left(1-1/\sqrt{2}\right)\lambda$-th leftmost agent therein, and then chooses the rightmost representative as the final winner. See Mechanism~\ref{mech:RoL} for a detailed description. This mechanism is an implementation of $p$-Statistic-of-$q$-Statistic with $p=k$ and $q=\left(1-1/\sqrt{2}\right)\lambda$, and is thus strategyproof. So, it suffices to show that it achieves a distortion of at most $1+\sqrt{2}$.

\SetCommentSty{mycommfont}
\begin{algorithm}[ht]
\SetNoFillComment
\caption{\sc Rightmost-of-$\left(1-1/\sqrt{2}\right)\lambda$-Leftmost}
\label{mech:RoL}
\For{each district $d\in D$}
{
    $y_d := \left(1-1/\sqrt{2}\right)\lambda$-th leftmost agent\;
}
\Return $w := $ rightmost representative\; 
\end{algorithm}

\begin{theorem}
For Max-of-Sum, the distortion of {\sc Rightmost-of-$\left(1-1/\sqrt{2}\right)\lambda$-leftmost} is at most $1+\sqrt{2}$. 
\end{theorem}

\begin{proof}
Let $w$ be the location chosen be the mechanism, and $o$ the optimal location. Denote by $d^*$ a district that gives the max sum for $w$, and by $d_w$ a district represented by $w$. Also, for any district $d$, we denote by $\cost_d(x) = \sum_{i \in N_d} \delta(i,x)$ the total distance of the agents in $d$ from location $x$, and let $o_d$ be the location that minimizes this distance (that is, $o_d$ is the median agent of $d$). Clearly, by definition, we have that $\cost(w) = \cost_{d^*}(w)$, and $\cost_d(o) \leq \cost(o)$ for every district $d$. We consider the following two cases:

\paragraph{Case 1: $o < w$.} \ \\
By the definition of $d^*$ and the triangle inequality, we have
\begin{align*}
\cost(w) = \sum_{i \in N_{d^*}} \delta(i,w) \leq \sum_{i \in N_{d^*}} \delta(i,o) + \sum_{i \in N_{d^*}} \delta(o,w) \leq \cost(o) + \lambda \delta(o,w). 
\end{align*}
Let $S = \{i\in N_{d_w}: x_i \geq w\}$ be the set of agents that are positioned at the right of (or exactly at) $w$ in $d_w$. 
By the definition of $w$, $|S| \geq \frac{1}{\sqrt{2}}\lambda$. Since $o < w$, we have 
\begin{align*}
\cost_{d_w}(o) \geq |S| \cdot \delta(w,o) \geq \frac{1}{\sqrt{2}}\lambda \cdot \delta(w,o)
\Leftrightarrow 
\lambda \delta(w,o) \leq \sqrt{2} \cdot \cost_{d_w}(o) \leq \sqrt{2} \cdot \cost(o).
\end{align*}
By combining everything together, we obtain a bound of $1+\sqrt{2}$.

\paragraph{Case 2: $w < o$.} \ \\
We consider the following two subcases:
\begin{itemize}
\item  $o_{d^*} \leq w < o$. By the monotonicity of the social cost\footnote{It is a well-known fact that the social cost objective is monotone in the locations. In particular, for any set of agents $S$, if $y_1 \in \arg\min_{x} \sum_{i \in S} \delta(i,x)$, then $\sum_{i \in S} \delta(i,y_1) \leq \sum_{i \in S} \delta(i,y_2) \leq \sum_{i \in S} \delta(i,y_3)$ for any $y_1 \leq y_2 \leq y_3$ or $y_3 \leq y_2 \leq y_1$.} for the agents in district $d^*$, we have that $\cost_{d^*}(o_{d^*}) \leq \cost_{d^*}(w) \leq \cost_{d^*}(o)$, and thus $\cost(w) \leq \cost(o)$. 

\item $w < o_{d^*}$. Since $w$ is the rightmost representative, it must be the case that $y_{d^*} \leq w < o_{d^*}$. So, again by the monotonicity of the social cost within the district $d^*$, we have that $\cost_{d^*}(o_{d^*}) \leq \cost_{d^*}(w) \leq \cost_{d^*}(y_{d^*})$. We will now argue that $\cost_{d^*}(y_{d^*}) \leq (1+\sqrt{2})\cost_{d^*}(o_{d^*})$. 
Let $L$ be the set that includes $(1-\frac{1}{\sqrt{2}})\lambda$ agents of $d^*$ from the leftmost to the $(1-\frac{1}{\sqrt{2}})\lambda$-th leftmost agent (that is, $y_{d^*}$), and the set $R$ that includes the remaining agents. By definition, we have that $|R|/|L|=1+\sqrt{2}$. Now, observe that
\begin{itemize}
    \item For every agent $i \in L$, $i \leq y_{d^*}$, and thus $\delta(i,o_{d^*}) = \delta(i,y_{d^*}) + \delta(y_{d^*},o_{d^*})$.
    \item For every agent $i \in R$, $i \geq y_{d^*}$, and thus $\delta(i,y_{d^*}) \leq \delta(i,o_{d^*}) + \delta(y_{d^*},o_{d^*})$.
\end{itemize}
Hence, 
\begin{align*}
\cost_{d^*}(y_{d^*}) &= \sum_{i \in N_{d^*}} \delta(i,y_{d^*}) = \sum_{i \in L} \delta(i,y_{d^*}) + \sum_{i \in R} \delta(i,y_{d^*}) \\
&\leq \sum_{i \in L} \delta(i,y_{d^*}) + \sum_{i \in R} \bigg( \delta(i,o_{d^*}) + \delta(y_{d^*},o_{d^*}) \bigg) \\
&= \sum_{i \in L} \bigg( \delta(i,y_{d^*}) + \delta(y_{d^*},o_{d^*}) \bigg) + \sum_{i \in R} \delta(i,o_{d^*})  +(|R|-|L|) \delta(y_{d^*},o_{d^*}) \\
&= \cost_{d^*}(o_{d^*}) + (|R|-|L|) \delta(y_{d^*},o_{d^*}).
\end{align*}
Since $y_{d^*} < o_{d^*}$, we also have that $\cost_{d^*}(o) \geq |L| \delta(y_{d^*},o_{d^*})$, and thus
\begin{align*}
\cost_{d^*}(y_{d^*}) &\leq \cost_{d^*}(o_{d^*}) + \frac{|R|-|L|}{|L|} \cost_{d^*}(o_{d^*}) = \frac{|R|}{|L|}\cost_{d^*}(o_{d^*}) = (1 + \sqrt{2})\cost_{d^*}(o_{d^*}). 
\end{align*}
From this, we finally get that $\cost_{d^*}(w) \leq (1+\sqrt{2}) \cost_{d^*}(o_{d^*}) \leq (1+\sqrt{2}) \cost_{d^*}(o)$, and thus $\cost(w) \leq (1+\sqrt{2}) \cost(o)$.
\hfill $\qedhere$
\end{itemize}
\end{proof}

\section{Open problems}
In this paper we settled the distortion of unrestricted and strategyproof mechanisms for the distributed single-facility location problem in terms of social objectives that are combinations of sum and max. There are several interesting directions for future work, such as to extend our work to more general metric spaces or to define further meaningful objectives and study similar questions about efficiency and strategyproofness. Beyond the single-facility location problem that we studied here, one could consider settings with more facilities and agents that have heterogeneous preferences over the facilities.

\bibliographystyle{plainnat}
\bibliography{references}

\end{document}